\documentclass[pra,aps,twocolumn,showpacs]{revtex4}

\usepackage{amsmath,amsfonts,amssymb,caption,hyperref,color,epsfig,graphics,graphicx,latexsym,mathrsfs,revsymb,theorem,url,epstopdf}

\hypersetup{colorlinks,linkcolor={blue},citecolor={blue},urlcolor={red}}

\usepackage{hyperref}

\newlength{\halfpagewidth}

\divide\halfpagewidth by 2

\newtheorem{definition}{Definition}
\newtheorem{proposition}[definition]{Proposition}
\newtheorem{Lemma}[definition]{Lemma}

\newtheorem{Theorem}[definition]{Theorem}
\newtheorem{Corollary}[definition]{Corollary}
\newtheorem{conjecture}[definition]{Conjecture}

\newtheorem{remark}[definition]{Remark}
\newtheorem{example}[definition]{Example}
\newtheorem{question}[definition]{Question}

\def\squareforqed{\hbox{\rlap{$\sqcap$}$\sqcup$}}
\def\qed{\ifmmode\squareforqed\else{\unskip\nobreak\hfil
		\penalty50\hskip1em\null\nobreak\hfil\squareforqed
		\parfillskip=0pt\finalhyphendemerits=0\endgraf}\fi}
\def\endenv{\ifmmode\;\else{\unskip\nobreak\hfil
		\penalty50\hskip1em\null\nobreak\hfil\;
		\parfillskip=0pt\finalhyphendemerits=0\endgraf}\fi}
\newenvironment{proof}{\noindent \textbf{{Proof.~} }}{\qed}
\def\Dbar{\leavevmode\lower.6ex\hbox to 0pt
	{\hskip-.23ex\accent"16\hss}D}
\makeatletter
\def\url@leostyle{%
	\@ifundefined{selectfont}{\def\UrlFont{\sf}}{\def\UrlFont{\small\ttfamily}}}
\makeatother
\def\bcj{\begin{conjecture}}
	\def\ecj{\end{conjecture}}
\def\bcr{\begin{corollary}}
	\def\ecr{\end{corollary}}
\def\bd{\begin{definition}}
	\def\ed{\end{definition}}
\def\bea{\begin{eqnarray}}
\def\eea{\end{eqnarray}}
\def\bem{\begin{enumerate}}
	\def\eem{\end{enumerate}}
\def\bex{\begin{example}}
	\def\eex{\end{example}}
\def\bim{\begin{itemize}}
	\def\eim{\end{itemize}}
\def\bl{\begin{lemma}}
	\def\el{\end{lemma}}
\def\bma{\begin{bmatrix}}
	\def\ema{\end{bmatrix}}
\def\bpf{\begin{proof}}
	\def\epf{\end{proof}}
\def\bpp{\begin{proposition}}
	\def\epp{\end{proposition}}
\def\bqu{\begin{question}}
	\def\equ{\end{question}}
\def\br{\begin{remark}}
	\def\er{\end{remark}}
\def\bt{\begin{theorem}}
	\def\et{\end{theorem}}

\def\btb{\begin{tabular}}
	\def\etb{\end{tabular}}

\newcommand{\nc}{\newcommand}

\nc{\bbA}{\mathbb{A}} \nc{\bbB}{\mathbb{B}} \nc{\bbC}{\mathbb{C}}
\nc{\bbD}{\mathbb{D}} \nc{\bbE}{\mathbb{E}} \nc{\bbF}{\mathbb{F}}
\nc{\bbG}{\mathbb{G}} \nc{\bbH}{\mathbb{H}} \nc{\bbI}{\mathbb{I}}
\nc{\bbJ}{\mathbb{J}} \nc{\bbK}{\mathbb{K}} \nc{\bbL}{\mathbb{L}}
\nc{\bbM}{\mathbb{M}} \nc{\bbN}{\mathbb{N}} \nc{\bbO}{\mathbb{O}}
\nc{\bbP}{\mathbb{P}} \nc{\bbQ}{\mathbb{Q}} \nc{\bbR}{\mathbb{R}}
\nc{\bbS}{\mathbb{S}} \nc{\bbT}{\mathbb{T}} \nc{\bbU}{\mathbb{U}}
\nc{\bbV}{\mathbb{V}} \nc{\bbW}{\mathbb{W}} \nc{\bbX}{\mathbb{X}}
\nc{\bbZ}{\mathbb{Z}}

\nc{\bA}{{\bf A}} \nc{\bB}{{\bf B}} \nc{\bC}{{\bf C}}
\nc{\bD}{{\bf D}} \nc{\bE}{{\bf E}} \nc{\bF}{{\bf F}}
\nc{\bG}{{\bf G}} \nc{\bH}{{\bf H}} \nc{\bI}{{\bf I}}
\nc{\bJ}{{\bf J}} \nc{\bK}{{\bf K}} \nc{\bL}{{\bf L}}
\nc{\bM}{{\bf M}} \nc{\bN}{{\bf N}} \nc{\bO}{{\bf O}}
\nc{\bP}{{\bf P}} \nc{\bQ}{{\bf Q}} \nc{\bR}{{\bf R}}
\nc{\bS}{{\bf S}} \nc{\bT}{{\bf T}} \nc{\bU}{{\bf U}}
\nc{\bV}{{\bf V}} \nc{\bW}{{\bf W}} \nc{\bX}{{\bf X}}
\nc{\bZ}{{\bf Z}}

\nc{\cA}{{\cal A}} \nc{\cB}{{\cal B}} \nc{\cC}{{\cal C}}
\nc{\cD}{{\cal D}} \nc{\cE}{{\cal E}} \nc{\cF}{{\cal F}}
\nc{\cG}{{\cal G}} \nc{\cH}{{\cal H}} \nc{\cI}{{\cal I}}
\nc{\cJ}{{\cal J}} \nc{\cK}{{\cal K}} \nc{\cL}{{\cal L}}
\nc{\cM}{{\cal M}} \nc{\cN}{{\cal N}} \nc{\cO}{{\cal O}}
\nc{\cP}{{\cal P}} \nc{\cQ}{{\cal Q}} \nc{\cR}{{\cal R}}
\nc{\cS}{{\cal S}} \nc{\cT}{{\cal T}} \nc{\cU}{{\cal U}}
\nc{\cV}{{\cal V}} \nc{\cW}{{\cal W}} \nc{\cX}{{\cal X}}
\nc{\cZ}{{\cal Z}}

\nc{\hA}{{\hat{A}}} \nc{\hB}{{\hat{B}}} \nc{\hC}{{\hat{C}}}
\nc{\hD}{{\hat{D}}} \nc{\hE}{{\hat{E}}} \nc{\hF}{{\hat{F}}}
\nc{\hG}{{\hat{G}}} \nc{\hH}{{\hat{H}}} \nc{\hI}{{\hat{I}}}
\nc{\hJ}{{\hat{J}}} \nc{\hK}{{\hat{K}}} \nc{\hL}{{\hat{L}}}
\nc{\hM}{{\hat{M}}} \nc{\hN}{{\hat{N}}} \nc{\hO}{{\hat{O}}}
\nc{\hP}{{\hat{P}}} \nc{\hR}{{\hat{R}}} \nc{\hS}{{\hat{S}}}
\nc{\hT}{{\hat{T}}} \nc{\hU}{{\hat{U}}} \nc{\hV}{{\hat{V}}}
\nc{\hW}{{\hat{W}}} \nc{\hX}{{\hat{X}}} \nc{\hZ}{{\hat{Z}}}

\nc{\hn}{{\hat{n}}}

\def\max{\mathop{\rm max}}
\def\min{\mathop{\rm min}}

\def\rank{\mathop{\rm rank}}

\def\tr{\mathop{\rm Tr}}

\newcommand{\bra}[1]{\langle#1|}
\newcommand{\ket}[1]{|#1\rangle}

\newcommand{\norm}[1]{\lVert#1\rVert}

\def\Dbar{\leavevmode\lower.6ex\hbox to 0pt
	{\hskip-.23ex\accent"16\hss}D}

\begin{document}
	\title{Entanglement polygon inequalities for pure states in qudit systems}
	
	\author{Xian Shi}\email[]
	{shixian01@gmail.com}
\affiliation{College of Information Science and Technology, Beijing University of Chemical Technology, Beijing 100029, China}
	
	%
	
	
	
	\date{\today}
	

\begin{abstract}
\indent Entanglement is one of the important resources in quantum tasks. Recently, Yang $et$ $al.$ [arXiv:2205.08801] proposed an entanglement polygon inequalities (EPI) in terms of some entanglement measures for $n$-qudit pure states. Here we continue to consider the entanglement polygon inequalities. Specifially, we show that the EPI is valid for $n$-qudit pure states in terms of geometrical entanglement measure (GEM), then we study the residual entanglement in terms of GEM for pure states in three-qubit systems. At last, we present counterexamples showing that the EPI is invalid for higher dimensional systems in terms of negativity, we also present a class of states beyond qubits satisfy the EPI in terms of negativity. 
\end{abstract}

\maketitle
\section{introduction}
\indent Quantum entanglement is an essential feature of quantum mechanics. It plays an important role in quantum informa- tion and quantum computation theory \cite{horodecki2009quantum}, such as superdense coding \cite{bennett1992communication}, teleportation \cite{bennett1993teleporting}, and the speedup of quantum algorithms \cite{shimoni2005entangled}.\\
\indent One of the important properties of multipartite entanglement is that entanglement cannot be freely shared. For a tripartite entangled state $\rho_{ABC},$ there exists six different bipartite entanglements $E_{A|B},E_{A|C},E_{B|C},$$E_{A|BC},E_{B|AC}$ and $E_{C|AB}$. Here $E$ is an entanglement measure for bipartite systems, $E_{A|BC}$ can be seen as a $one$ $to$ $group$ entanglement, and $E_{A|B}$ can be seen as a $one$ $to$ $one$ entanglement. In 2000, Coffman, Kundu and Wootters showed a famous inequality for a three qubit system $\mathcal{H}_A\otimes\mathcal{H}_B\otimes\mathcal{H}_C$ \cite{coffman2000distributed},  
\begin{align}
E_{A|BC}\ge E_{A|B}+E_{A|C},\label{moe}
\end{align}
 This inequality means that the entanglement between a singled out qubit to a group of qubits is bounded by the sum of the entanglement between the singled out qubit to each qubit in the group. Later, this inequality is generalized to multiqubit systems in terms of the squashed entanglent measure \cite{christandl2004squashed}, the squared entanglement of formation \cite{bai2014general}, the squared Tsallis-$q$ entanglement measure \cite{luo2016general}, and the squared Renyi entanglement measure \cite{song2016general}. As (\ref{moe}) is not valid for multipartite higher dimensional systems in terms of generic bipartite entanglement measures, other types of monogamy of entanglement were presented \cite{de2014monogamy,gour2018monogamy,shi2021multilinear}. \par
 Recently, entanglement polygon inequality (EPI) that is on the $one$ $to$ $group$ entanglement was proposed for multipartite entangled pure states $\otimes_i\mathcal{H}_i$ \cite{qian2018entanglement,yang2022entanglement},
 \begin{align}
 E_{j|\overline{j}}\le \sum_{k\ne j}E_{k|\overline{k}}.
 \end{align}
 This inequality was proved valid for pure states in multi-qubit systems in terms of bipartite entanglement measures \cite{qian2018entanglement} and in arbitrary dimensional systems in terms of the $q$-concurrence \cite{yang2021parametrized} when $q\ge 2$ and the unified-$(q,s)$ entangled measure when $q\ge 1, s\ge0$ \cite{yang2022entanglement}. However, whether the EPI in valid for pure states in higher dimensional systems in terms of concurrence or negativity is unknown.\par
This paper is organized as follows. In Sec.  \MakeUppercase{\romannumeral2}, we present the preliminary knowledge needed here, In Sec.  \MakeUppercase{\romannumeral3}, we present our main results, first we present that the EPI is valid in terms of GEM for pure states in arbitrary dimensional systems, we also consider the entanglement indicators based the EPI. At last, We present a class of pure states that doesnot satisfy the EPI in terms of negativity. Moreover, we present a class of pure states in higher dimensional systems satisfy the EPI  in terms of negativity, we also present that the EPI is valid in terms of concurrence for pure states in arbitrary dimensional systems. In Sec.  \MakeUppercase{\romannumeral4}, we end with a summary.
\section{PRELIMINARIES}
\indent An $n$-partite pure state $\ket{\psi}_{A_1A_2\cdots A_n}$ is full product if it can be written as
\begin{align}
\ket{\psi}_{A_1A_2\cdots A_n}=\ket{\phi_1}_{A_1}\ket{\phi_2}_{A_2}\cdots \ket{\phi_n}_{A_n},
\end{align}
 otherwise, it is entangled. A multipartite pure state is called
genuinely entangled if
\begin{align}
\ket{\psi}_{A_1A_2\cdots A_n}\ne \ket{\phi}_S\ket{\zeta}_{\overline{S}},
\end{align}
for any partite $S|\overline{S}$, here $S$ is a subset of $\boldsymbol{A}=\{A_1,A_2,\cdots, A_n\}$, and $\overline{S}=\boldsymbol{A}-S$.\par
Next we recall some entanglement measures for a bipartite state $\rho_{AB}.$ A bipartite pure state $\ket{\psi}_{AB}$ can be always written as 
\begin{align*}
\ket{\psi}_{AB}=\sum_{i=0}^{d-1}\sqrt{\lambda_i}\ket{i}_A\ket{i}_B,
\end{align*}
where $\lambda_i\ge \lambda_{i+1}\ge 0,$ $i=0,1,\cdots,d-2,$ and $\sum_i\lambda_i=1$.  The geometric entanglement measure for a pure state $\ket{\phi}_{AB}=\sum_i\sqrt{\lambda_i}\ket{ii}$ is defined as
\begin{align}
G(\ket{\phi}):=&1-\max_{\ket{\psi}=\ket{a_1}\ket{a_2}}\bra{\psi}\rho\ket{\psi}\label{gd}
\end{align}
The geometric entanglement measure for a mixed state $\rho_{AB}$ is generally defined as \cite{wei2003geometric},
\begin{align}
G(\rho_{AB})=\min\sum_ip_i G(\ket{\phi_i}),
\end{align}
where the minimum takes over all the decompositions of $\rho_{AB}$. \par
Negativity is an important entanglement measure for bipartite systems. It is defined as follows for a pure state $\ket{\phi}_{AB}$ \cite{vidal2002computable}, 
\begin{align}
N(\ket{\phi}_{AB})=\frac{\norm{\ket{\psi}\bra{\psi}^{T_A}}_1-1}{2},\label{nd}
\end{align}
here $\bra{ij}\rho_{AB}^{T_A}\ket{kl}=\bra{kj}\rho_{AB}\ket{il}$.\par
Concurrence is the other entanglement measure for bipartite mixed states $\rho_{AB}$. The concurrence of a pure state $\ket{\psi}_{AB}$ is defined as 
\begin{align}
C(\ket{\psi}_{AB})=\sqrt{2(1-\tr\rho_A^2)}.\label{c}
\end{align}
Here $\rho_A=\tr_B\ket{\psi}_{AB}\bra{\psi}$. For a mixed state $\rho_{AB},$ it is defined as
\begin{align}\label{Cm}
C(\rho_{AB})=\min_{\{p_i,\ket{\phi_i}_{AB}\}}\sum_i p_i C(\ket{\phi_i}_{AB}),
\end{align}
where the minimum takes over all the decompositions of $\rho_{AB}=\sum_i p_i\ket{\phi_i}_{AB}\bra{\phi_i}$ with $p_i\ge 0$ and $\sum p_i=1.$\\
\indent For a two-qubit mixed state $\rho_{AB},$ Wootters derived an analytical formula \cite{wootters1998entanglement}:
\begin{align}
C_{AB}=&\max\{\sqrt{\mu_1}-\sqrt{\mu_2}-\sqrt{\mu_3}-\sqrt{\mu_4},0\},\label{C2}
\end{align}
where $\mu_1,\mu_2,\mu_3,$ and $\mu_4,$ are the eigenvalues of the matrix $\rho_{AB}(\sigma_y\otimes\sigma_y)\rho_{AB}^{*}(\sigma_y\otimes\sigma_y)$ with nonincreasing order.\par
Recently, a parametrized entanglement measure, the $q$-concurrence, was proposed in \cite{yang2021parametrized}. The $q$-concurrence for a pure state $\ket{\psi}_{AB}$ is defined as 
\begin{align}
C_q(\ket{\psi}_{AB})=1-Tr\rho_A^q. \label{qc}
\end{align}\par
In \cite{yang2022entanglement}, the authors showed that the $q$-concurrence satisfies the EPI for pure states when $q\ge 2$.
\begin{Lemma}\cite{yang2022entanglement}\label{l0}
	For any $n$-qudit pure entangled state $\ket{\psi}$ in $\otimes_i\mathcal{H}_i.$ When $q\ge 2$, the following inequality holds,
	\begin{align}
	C_q^{j|\overline{j}}\le\sum_{k\ne j,\forall k}C_q^{k|\overline{k}}(\ket{\psi}),
	\end{align}
\end{Lemma}	
\section{Main Results}
\indent In this section, we first recall the EPI proposed in \cite{,qian2018entanglement,yang2022entanglement}, then we present that the geometric entanglement measure (GEM) also satisfies the EPI for three qudit pure states, and we also present an application based on the EPI in terms of GEM.  At last, we present a class of  counterexamples which shows that negativity is invalid for three qudit pure states, we also present a class of examples that satisfies the EPI in terms of negativity.
\subsection{EPI IN TERMS OF GEOMETRIC ENTANGLEMENT MEASURE}
\indent  Assume $\ket{\psi}_{AB}=\sum_{i=0}^{d-1}\sqrt{\lambda_i}\ket{ii}$ is a pure state with $\lambda_i\ge \lambda_{i+1}\ge 0,$ $i=0,1,\cdots,d-2,$ and $\sum_i\lambda_i=1$. By the formula (\ref{gd}), we have 
\begin{align}
G(\ket{\psi})=1-\lambda_0, \label{gm}
\end{align}
Next we recall the Schatten $p$-norm of a bounded operator $M$ on the Hilber space $\mathcal{H}$ with finite dimensions,
\begin{align}
\norm{M}_p:=&[\tr(|M|^p)]^{\frac{1}{p}},\nonumber\\
=&(\sum_{i=0}^{d-1}s_i^p(M))^{\frac{1}{p}}.
\end{align}
here $s_i(M)$ are the sigular values of $M$ with decreasing order, $i.$$e.$ the eigenvalues of $\sqrt{M^{\dagger}M}$.  And when $p\rightarrow\infty,$ $\norm{M}_p\rightarrow s_0(M).$ Next we recall the following lemma.
\begin{Lemma}\cite{audenaert2007subadditivity}
	For any bipartite state $\rho$ on a Hilbert space $\mathcal{H}_1\otimes\mathcal{H}_2$, the inequality 
	\begin{align}
	1+\norm{\rho}_q\ge \norm{Tr_1\rho}_q+\norm{Tr_2\rho}_1 \label{sq}
	\end{align}
	holds for $q>1$.
\end{Lemma}
Clearly, the inequality $(\ref{sq})$ can be written as 
\begin{align}
\sum_{i=1}^21-\norm{Tr_i\rho}_q\ge1-\norm{\rho}_q,\label{sq1}\hspace{5mm}\textit{when $q>1.$}
\end{align}
 Then we have the following result.
\begin{Theorem}\label{th1}
Assume $\ket{\psi}_{A_1A_2\cdots A_n}$ is a pure state in the following system $\otimes_{i=1}^n\mathcal{H}_i$, then we have the following inequality,
\begin{align}
G(\ket{\psi}_{P_i|\overline{P_i}})\le \sum_{j\ne i}G(\ket{\psi}_{P_j|\overline{P_j}}),
\end{align}
here $\{P_1,P_2,\cdots,P_k\}$ is a partition of the set $\{A_1,A_2,\cdots,A_n\}$ with $k\le n$.
\end{Theorem}
\begin{proof}As
\begin{align}
G(\ket{\psi}_{AB})=1-s_0(\rho_A)=\lim_{p\rightarrow\infty}[1-\norm{\rho_A}_p], \label{ge}
\end{align}
then we have
\begin{align}
G(\ket{\psi}_{P_i|\overline{P_i}})=&\lim_{p\rightarrow\infty}[1-\norm{\rho_{P_i}}_p]\nonumber\\
=&\lim_{p\rightarrow\infty}[1-\norm{\rho_{\cup_{j\ne i} P_j}}_p]\nonumber\\
\le&\lim_{p\rightarrow\infty}\sum_{j\ne i} [1-\norm{\rho_{P_j}}_q]\nonumber\\
=&\sum_{j\ne i}G(\ket{\psi}_{P_j|\overline{P_j}}).
\end{align}
Here the first equality is due to $(\ref{ge}),$ the first inequality is due to the inequality $(\ref{sq1})$, and the last inequality is due to the equality $(\ref{ge}).$
\end{proof}\par
Next we propose a generalized EPI for $n$-qudit pure states in terms of the GEM.
\begin{Lemma}\label{l1}
	When $a,b,c\in (0,1]$, and $a+b\ge c$, then 
	\begin{align}
	a^{\alpha}+b^{\alpha}\ge c^{\alpha},\nonumber
	\end{align}
	when $\alpha\in (0,1]$.
\end{Lemma}
\begin{proof}
Here we can always assume $a\ge b$, and when $0<\alpha\le 1,$ $(\frac{a}{b}+1)^{\alpha}-(\frac{a}{b})^{\alpha}\le 1,$ then we have
\begin{align*}
a^{\alpha}+b^{\alpha}\ge (a+b)^{\alpha}\ge c^{\alpha}.
\end{align*}
\end{proof}
\begin{Corollary}\label{gepi}
	Assume $\ket{\psi}_{A_1A_2\cdots A_n}$ is a pure state in the following system $\otimes_{i=1}^n\mathcal{H}_i$, then we have the following inequality,
	\begin{align}
	G(\ket{\psi}_{P_i|\overline{P_i}})^{\alpha}\le \sum_{j\ne i}G(\ket{\psi}_{P_j|\overline{P_j}})^{\alpha},\label{gm2}
	\end{align}
	here $\{P_1,P_2,\cdots,P_k\}$ is a partition of the set $\{A_1,A_2,\cdots,A_n\}$ with $k\le n$, and $\alpha\in (0,1]$.
\end{Corollary}
The above corollary is due to the Theorem \ref{th1} and Lemma \ref{l1}. Then we present an example for the above corollary.
\begin{example}
\begin{align*}
\ket{\psi}_{ABC}=\frac{3}{5}\ket{102}+\frac{2\sqrt{2}}{5}\ket{200}+\frac{2}{5}\ket{010}+\frac{\sqrt{2}}{5}\ket{020}+\frac{\sqrt{2}}{5}\ket{001},
\end{align*}
through computation, we have 
\begin{align}
G(\ket{\psi}_{A|BC})=\frac{9}{25}, G(\ket{\psi}_{B|AC})=\frac{19}{25}, G(\ket{\psi}_{C|AB})=\frac{14}{25},
\end{align}
then the inequality $(\ref{gm2})$ can be written as
\begin{align}
g=&G(\ket{\psi}_{A|BC})^{\alpha}+G(\ket{\psi}_{C|BA})^{\alpha}-G(\ket{\psi}_{B|AC})^{\alpha}\nonumber\\
=&(\frac{9}{25})^{\alpha}+(\frac{14}{25})^{\alpha}-(\frac{19}{25})^{\alpha},\label{e1}
\end{align}
we plot $(\ref{e1})$ in Fig. \ref{fig1}, from the figure, we see that the inequality $(\ref{e1})$ is bigger than 0 when $\alpha\in (0,1).$
\begin{figure}
	\centering
	\includegraphics[width=90mm]{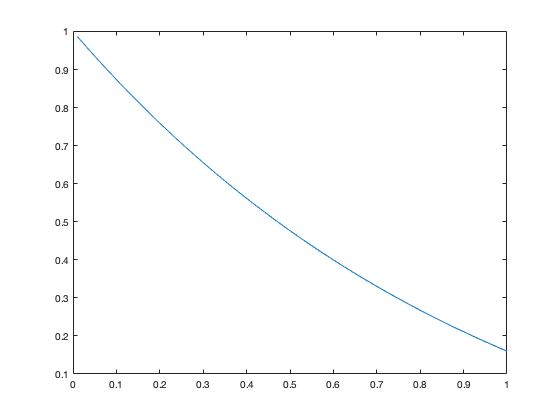}\\
	\caption{In this figure,  we present the inequality (\ref{e1}).  }\label{fig1}
\end{figure}
\end{example}
\subsection{APPLICATIONS OF THE EPI IN TERMS OF GEM}
\indent Assume $\ket{\psi}_{A_1A_2\cdots A_n}$ is a pure state in $\otimes_i\mathcal{H}_i$,  then we consider an entanglement indicator defined as \cite{zhu2014entanglement,yang2022entanglement},
\begin{align}
\delta_{\alpha}^{G}(\ket{\psi})=&\min_{i} \tau_{\alpha}^i(\ket{\psi}),\nonumber\\
\tau_{\alpha}^i(\ket{\psi})=&\sum_{j\ne i}G^{\alpha}(\ket{\psi}_{A_j|\overline{A_j}})-G^{\alpha}(\ket{\psi}_{A_i|\overline{A_i}}).
\end{align} \par
By the Corollary \ref{gepi}, we have that $\delta_{\alpha}^{G}\ge 0,$ $\alpha\in(0,1)$. When considering $\ket{\psi}$ is a three-qubit pure state, we have the following theorem.
\begin{Theorem}
	Assume $\ket{\psi}_{A_1A_2A_3}$ is a three-qubit pure state, then $\delta_{\alpha}^{G}=0,$ $\forall\alpha\in(0,1)$ if and only if $\ket{\psi}_{A_1A_2A_3}$ is a  biseparable state.
\end{Theorem}
\begin{proof}
	$\Rightarrow:$\hspace{4mm}
	When $\alpha\in(0,0.5)$, as $2\alpha\in(0,1),$ we have \begin{align*}
	&G^{2\alpha}(\ket{\psi}_{i|\overline{i}})-G^{2\alpha}(\ket{\psi}_{j|\overline{j}})-G^{2\alpha}(\ket{\psi}_{k|\overline{k}})\nonumber\\
	=&(G^{\alpha}(\ket{\psi}_{j|\overline{j}})+G^{\alpha}(\ket{\psi}_{k|\overline{k}}))^2-G^{2\alpha}(\ket{\psi}_{j|\overline{j}})-G^{2\alpha}(\ket{\psi}_{k|\overline{k}})\nonumber\\
	=&2G^{\alpha}(\ket{\psi}_{j|\overline{j}})G^{\alpha}(\ket{\psi}_{k|\overline{k}})=0.
	\end{align*}
	Here $j\ne k\ne i.$ And
	From the above, we have $G(\ket{\psi}_{j|\overline{j}})=0$ or $G(\ket{\psi}_{k|\overline{k}})=0.$\par
	A three-qubit pure state $\ket{\psi}_{A_1A_2A_3}$ can be written in the generalized Schmidt decomposition \cite{acin2000generalized}:
	\begin{align}
	\ket{\psi}=l_0\ket{000}+l_1e^{i\theta}\ket{100}+l_2\ket{101}+l_3\ket{110}+l_4\ket{111},
	\end{align}
	where $\theta\in[0,\pi)$, $l_i\ge 0$ $(i=0,1,2,3,4),$ and $\sum_{i=0}^4l_i^2=1.$ From simple computation, we have the Schmidt coefficients of $\ket{\psi}$ in terms of $A|BC$ is $$Sch(\ket{\psi}_{A|BC})=(l_0^2,1-l_0^2),$$  the Schmidt coefficients of $\ket{\psi}$ in terms of $B|AC$ is $$Sch(\ket{\psi}_{B|AC})=(\frac{1+\sqrt{1-4\Delta_0}}{2},\frac{1-\sqrt{1-4\Delta_0}}{2}),$$ here $\Delta_0=l_0^2l_3^2+l_0^2l_4^2+l_1^2l_4^2+l_2^2l_3^2-2l_1l_2l_3l_4\cos\theta,$ the Schmidt coefficients of $\ket{\psi}$ in terms of $C|AB$ is 
	\begin{align*}
	Sch(\ket{\psi}_{C|AB})=(\frac{1+\sqrt{1-4\Delta_1}}{2},\frac{1-\sqrt{1-4\Delta_1}}{2}),
	\end{align*}
	here $\Delta_1=l_0^2l_2^2+l_0^2l_4^2+l_1^2l_4^2+l_2^2l_3^2-2l_1l_2l_3l_4\cos\theta.$ \par
	When $G(\ket{\psi}_{A|BC})=0,$ $l_0=0$ or $l_0=1,$ then $\ket{\psi}_{ABC}$ is biseparable. When $G(\ket{\psi}_{B|AC})=0,$ that is, $\Delta_0=0,$ as $\Delta_0=l_0^2(l_3^2+l_4^2)+(l_1l_4-l_2l_3)^2+4l_1l_2l_3l_4\sin^2\frac{\theta}{2}=0,$ then $l_3=l_4=0,$ or $l_0=0,$ then $\ket{\psi}_{ABC}$ is biseparable. The case $G(\ket{\psi}_{C|AB})=0$ is similar to $G(\ket{\psi}_{B|AC})=0$.  \par
	$\Leftarrow:$ \hspace{4mm} Assume $\ket{\psi}=\ket{\phi}_{i}\ket{\varphi}_{jk},$ $i\ne j\ne k,$ then $G(\ket{\psi}_{i|jk})=0,$ $G(\ket{\psi}_{j|ki})=G(\ket{\psi}_{k|ji}),$ then $\delta_{\alpha}^G=0,$ $\forall\alpha\in(0,1).$
\end{proof}
\subsection{EPI IN TERM OF NEGATIVITY AND CONCURRENCE}
\indent In this subsection, we present a class of counterexamples that does not satisfy the EPI in terms of negativity and concurrence, which answers the problem proposed in \cite{yang2022entanglement}. We also propose a class of multipartite pure states satisfying the EPI in terms of negativity.\par First we present an example.
\begin{example}
	\begin{widetext}
	\begin{align}
	\ket{\psi}_{ABC}
	=\frac{\ket{000}+\ket{101}+\ket{202}+\ket{310}+\ket{411}+\ket{512}+\ket{620}+\ket{721}+\ket{822}}{3},\label{r2}
	\end{align}\par
	According to the definition of Negativity (\ref{nd}), we have that $N(\ket{\psi}_{A|BC})=4,$ $N(\ket{\psi}_{B|AC})=1$, $N(\ket{\psi}_{C|AB})=1$. Clearly, $N(\ket{\psi}_{A|BC})\ge N(\ket{\psi}_{B|AC})+N(\ket{\psi}_{C|AB})$, that is, the EPI is invalid for $\ket{\psi}_{ABC}$ in terms of negativity. 
	\end{widetext}
\end{example}
\begin{remark}
	In the above example, we have that $\rho_{BC}=\frac{1}{9}I_3\otimes I_3$, that is, $\ket{\psi}_{ABC}$ is a purification of a product state $\rho_{BC}$.
 \end{remark}\par
Next we present a class of states which doesnot satisfy the EPI in terms of negativity.
\begin{Theorem}
	Assume $\rho$ and $\sigma$ are two states on the system $\mathcal{H}_d$, $\rho=\sum_i a_i\ket{i}\bra{i},$ $\sigma=\sum_j b_j\ket{j}\bra{j}$, $\rank(\rho),\rank(\sigma)>1,$ and $\{\ket{i}\}$ and $\{\ket{j}\}$ are the orthonormal bases of $\mathcal{H}_d.$  Let $\ket{\psi}_{ABC}=\sum_{ij}\sqrt{a_ib_j}\ket{ij}_{AB}\ket{ij}_C$ be a purification of $\rho_A\otimes\sigma_B$, then we have 
	\begin{align}
	N(\ket{\psi}_{C|AB})\ge N(\ket{\psi}_{A|BC})+N(\ket{\psi}_{BAC}).
	\end{align}
\end{Theorem}
\begin{proof}
	Assume $\ket{\phi}_{AB}=\sum_i\sqrt{\lambda_i}\ket{ii},$ then 
	\begin{align}
	N(\ket{\phi}_{AB})=&\frac{(\sum_i\sqrt{\lambda_i})^2-1}{2},\label{n}
	\end{align}
	Due to (\ref{n}), we have 
	\begin{align}
		&N(\ket{\psi}_{C|AB})- N(\ket{\psi}_{A|BC})-N(\ket{\psi}_{BAC})\nonumber\\
		=&\frac{1}{2}[(\sum_{ij}\sqrt{a_ib_j})^2-1-(\sum_i\sqrt{a_i})^2+1-(\sum_j \sqrt{b_j})^2+1]\nonumber\\
		=&\frac{1}{2}[(1-(\sum_i\sqrt{a_i})^2)(1-(\sum_j\sqrt{b_j})^2)].\label{nep}
	\end{align}
	As $\sum_ia_i=\sum_jb_j=1,$ and $\rank\rho,\rank\sigma>1,$ then $(\ref{nep})>0.$
	Then we finish the proof.
\end{proof}\par
At last, we present a class of pure states, the generalized W class (GW) states, satisfy the EPI in terms of negativity. These states were first studied on the problem of the monogamy relations in terms of concurrence for higher dimensional systems \cite{san2008generalized}. Recently, the general monogamy relations of this class of states attracts much attention of the revelent researchers \cite{shi2020monogamy,lai2021tighter}. \par
Now let us recall the definition of the GW states $\ket{W_n^d}$,
\begin{align}
\ket{W_n^d}_{A_1\cdots A_n}=\sum_{i=1}^{d}(a_{1i}\ket{i0\cdots 0}+\cdots+a_{ni}\ket{00\cdots i}),\label{p5}
\end{align} 
where we assume $\sum_{i=1}^{d}\sum_{j=1}^{n}|a_{ji}|^2=1.$ Next we present a lemma and then we show the main result of this subsection.\par
\begin{Lemma}\cite{san2008generalized}
	Assume $\ket{\psi}_{AB_1B_2\cdots B_{n-1}}$ is a GW state, then
	for an arbitrary partition $\{P_1,P_2,\cdots,P_m\}$ of the set $S=\{A,B_1B_2,\cdots,B_{n-1}\},$ the state $\ket{\psi}_{P_1P_2\cdots P_m}$ is also a GW state, here we assume $P_i\cap P_j=\emptyset$ $(i\ne j)$ and $\cup_i P_i=S, m\le n$.
\end{Lemma}
\begin{Theorem}\label{gw}
Assume $\ket{\psi}_{A_1\cdots A_n}$ is a GW state, and here we denote $\{P_1,P_2,P_3\}$ is a partition of the set $\{A_{1},A_{2},\cdots,A_{n}\},$ $n\ge 3,$ then we have
\begin{align}
N(\ket{\psi}_{P_1|P_2P_3})\le N(\ket{\psi}_{P_2|P_1P_3})+N(\ket{\psi}_{P_3|P_2P_1}).
\end{align}
\end{Theorem}
\begin{proof}
	Due to the above lemma, we have that $\ket{\psi}_{P_1P_2P_3}$ is also a GW state, that is, $\ket{\psi}_{P_1P_2P_3}$ can be written as 
	\begin{align}
	\ket{\psi}_{P_1P_2P_3}=\sum_{i=1}^d(a_{1i}\ket{i00}+a_{2i}\ket{0i0}+a_{3i}\ket{00i}),
	\end{align}
	through computation, we have that 
	\begin{align}
	N(\ket{\psi}_{P_1|P_2P_3})=\sqrt{\sum_i |a_{1i}|^2}\times\sqrt{\sum_i (|a_{2i}|^2+|a_{3i}|^2)},\nonumber\\
	N(\ket{\psi}_{P_2|P_1P_3})=\sqrt{\sum_i |a_{2i}|^2}\times\sqrt{\sum_i (|a_{1i}|^2+|a_{3i}|^2)},\nonumber\\
		N(\ket{\psi}_{P_3|P_1P_3})=\sqrt{\sum_i |a_{3i}|^2}\times\sqrt{\sum_i (|a_{2i}|^2+|a_{1i}|^2)},
	\end{align}
	let $a={\sum_i|a_{1i}|^2},$ $b={\sum_i|a_{2i}|^2},$ $c={\sum_i|a_{3i}|^2}$, as
	\begin{align*}
	\sqrt{a(b+c)}\le& \sqrt{ab+ac+2bc+2\sqrt{bc(a+c)(a+b)}},\nonumber\\
	=&\sqrt{[\sqrt{b(a+c)}+\sqrt{c(a+b)}]^2}.
	\end{align*}
	Then we finish the proof.
\end{proof}\par
Similar to the proof of Theorem \ref{gw} and Corollary \ref{gepi}, we can generalize the Theorem \ref{gw} to the following.
\begin{Corollary}\label{c1}
	Assume $\ket{\psi}_{A_1\cdots A_n}$ is a GW state, and here we denote $\{P_1,P_2,\cdots,P_k\}$ is a partition of the set $\{A_{1},A_{2},\cdots,A_{n}\},$ $n\ge 3,$ 
	then we have the following inequality,
	\begin{align}
	N(\ket{\psi}_{P_i|\overline{P_i}})^{\alpha}\le \sum_{j\ne i}N(\ket{\psi}_{P_j|\overline{P_j}})^{\alpha}, \label{n3}
	\end{align}
	here $\{P_1,P_2,\cdots,P_k\}$ is a partition of the set $\{A_1,A_2,\cdots,A_n\}$ with $k\le n$, and $\alpha\in (0,1]$.
\end{Corollary}
Then we present an example to show the meaning of Theorem $(\ref{gw}).$
\begin{example}
	Here we consider a tripartite pure state,
	\begin{align*}
	\ket{\psi}_{ABCD}=0.3\ket{0001}+0.4\ket{0020}+0.5\ket{0100}+\sqrt{0.5}\ket{1000},
	\end{align*}
	here we take $P_1=\{A\}, P_2=\{BC\},$ $P_3=\{D\},$ then
	\begin{align*}
	N(\ket{\psi}_{P_1|P_2P_3})=0.5, &N(\ket{\psi}_{P_2|P_1P_3})=\sqrt{0.2419},\nonumber\\ N(\ket{\psi}_{P_3|P_1P_2})=&\sqrt{0.0819},
	\end{align*}
	then the inequality $(\ref{n3})$ can be written as
	\begin{align}
	h=&N(\ket{\psi}_{P_2|P_1P_3})^{\alpha}+N(\ket{\psi}_{P_3|P_1P_2})^{\alpha}-N(\ket{\psi}_{P_1|P_2P_3})^{\alpha}\nonumber\\
	=&0.2419^{\frac{\alpha}{2}}+0.0819^{\frac{\alpha}{2}}-0.5^{\alpha},\label{e2}
	\end{align}
	\begin{figure}
		\centering
		\includegraphics[width=90mm]{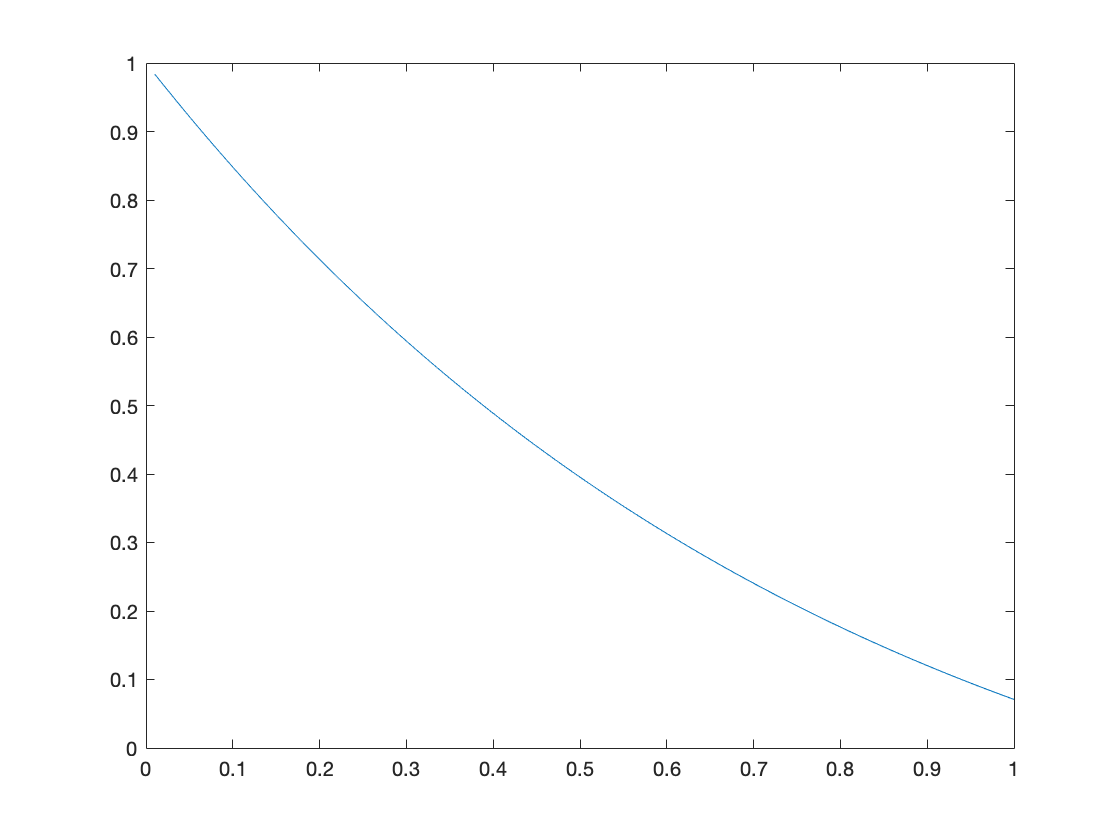}\\
		\caption{In this figure,  we present the inequality (\ref{e2}).  }\label{fig2}
	\end{figure}
	we plot $(\ref{e2})$ in Fig. \ref{fig2}, from the figure, we see that the inequality $(\ref{e1})$ is bigger than 0 when $\alpha\in (0,1).$
\end{example}
\par
At last, we present that the EPI is valid in terms of concurrence for pure states in arbitrary dimensional systems.

 Comparing the $(\ref{qc})$ with $(\ref{c})$, we have when $\ket{\psi}_{AB}$ a bipartite pure state, $C(\ket{\psi}_{AB})=\sqrt{2C_2(\ket{\psi}_{AB})}.$ Then combing the Lemma \ref{l0} and similar proof of Corollary \ref{gepi}, we have the following result,
 \begin{Theorem}
 		For any $n$-qudit pure entangled state $\ket{\psi}$ in $\otimes_i\mathcal{H}_i,$ we have
 	\begin{align}
 	C_{j|\overline{j}}\le\sum_{k\ne j,\forall k}C_{k|\overline{k}}(\ket{\psi}),
 	\end{align}
 \end{Theorem}
\section{CONCLUSION}
\indent In this article, we have investigated the EPI for pure states in multipartite systems. First we present the EPI is valid in terms of GEM for pure states in arbitrary dimensional systems. Based on the inequalities, we have presented the entanglement indicators in terms of GEM. At last, we have presented a class of pure states that does not satisfy the EPI in terms of negativity, we have also presented a class of pure states satisfying the EPI in terms of negativity. Moreover, we have presented the EPI is valid in terms of concurrence for pure states in arbitrary dimensional systems. Due to the importance of the study of the higher-dimensional multipartite entanglement systems, our results can provide a reference for future work on the study of multiparty quantum entanglement.
\bibliographystyle{IEEEtran}
\bibliography{ref}
\end{document}